\PassOptionsToPackage{unicode}{hyperref}
\PassOptionsToPackage{hyphens}{url}
\PassOptionsToPackage{dvipsnames,svgnames,x11names}{xcolor}
\documentclass[
  11pt,
  letterpaper,
]{article}
\usepackage{xcolor}
\usepackage[margin=1in]{geometry}
\usepackage{amsmath,amssymb}
\usepackage{iftex}
\ifPDFTeX
  \usepackage[T1]{fontenc}
  \usepackage[utf8]{inputenc}
  \usepackage{textcomp} % provide euro and other symbols
\else % if luatex or xetex
  \usepackage{unicode-math} % this also loads fontspec
  \defaultfontfeatures{Scale=MatchLowercase}
  \defaultfontfeatures[\rmfamily]{Ligatures=TeX,Scale=1}
\fi
\usepackage{lmodern}
\ifPDFTeX\else
\fi
\IfFileExists{upquote.sty}{\usepackage{upquote}}{}
\IfFileExists{microtype.sty}{% use microtype if available
  \usepackage[]{microtype}
  \UseMicrotypeSet[protrusion]{basicmath} % disable protrusion for tt fonts
}{}
\usepackage{setspace}
\makeatletter
\@ifundefined{KOMAClassName}{% if non-KOMA class
  \IfFileExists{parskip.sty}{%
    \usepackage{parskip}
  }{% else
    \setlength{\parindent}{0pt}
    \setlength{\parskip}{6pt plus 2pt minus 1pt}}
}{% if KOMA class
  \KOMAoptions{parskip=half}}
\makeatother
\makeatletter
\ifx\paragraph\undefined\else
  \let\oldparagraph\paragraph
  \renewcommand{\paragraph}{
    \@ifstar
      \xxxParagraphStar
      \xxxParagraphNoStar
  }
  \newcommand{\xxxParagraphStar}[1]{\oldparagraph*{#1}\mbox{}}
  \newcommand{\xxxParagraphNoStar}[1]{\oldparagraph{#1}\mbox{}}
\fi
\ifx\subparagraph\undefined\else
  \let\oldsubparagraph\subparagraph
  \renewcommand{\subparagraph}{
    \@ifstar
      \xxxSubParagraphStar
      \xxxSubParagraphNoStar
  }
  \newcommand{\xxxSubParagraphStar}[1]{\oldsubparagraph*{#1}\mbox{}}
  \newcommand{\xxxSubParagraphNoStar}[1]{\oldsubparagraph{#1}\mbox{}}
\fi
\makeatother

\usepackage{longtable,booktabs,array}
\usepackage{calc} % for calculating minipage widths
\usepackage{etoolbox}
\makeatletter
\patchcmd\longtable{\par}{\if@noskipsec\mbox{}\fi\par}{}{}
\makeatother
\IfFileExists{footnotehyper.sty}{\usepackage{footnotehyper}}{\usepackage{footnote}}
\makesavenoteenv{longtable}
\usepackage{graphicx}
\makeatletter
\newsavebox\pandoc@box
\newcommand*\pandocbounded[1]{% scales image to fit in text height/width
  \sbox\pandoc@box{#1}%
  \Gscale@div\@tempa{\textheight}{\dimexpr\ht\pandoc@box+\dp\pandoc@box\relax}%
  \Gscale@div\@tempb{\linewidth}{\wd\pandoc@box}%
  \ifdim\@tempb\p@<\@tempa\p@\let\@tempa\@tempb\fi% select the smaller of both
  \ifdim\@tempa\p@<\p@\scalebox{\@tempa}{\usebox\pandoc@box}%
  \else\usebox{\pandoc@box}%
  \fi%
}
\def\fps@figure{htbp}
\makeatother

\NewDocumentCommand\citeproctext{}{}
\NewDocumentCommand\citeproc{mm}{%
  \begingroup\def\citeproctext{#2}\cite{#1}\endgroup}
\makeatletter
 \let\@cite@ofmt\@firstofone
 \def\@biblabel#1{}
 \def\@cite#1#2{{#1\if@tempswa , #2\fi}}
\makeatother
\newlength{\cslhangindent}
\newlength{\csllabelwidth}
\newenvironment{CSLReferences}[2] % #1 hanging-indent, #2 entry-spacing
 {\begin{list}{}{%
  \setlength{\itemindent}{0pt}
  \setlength{\leftmargin}{0pt}
  \setlength{\parsep}{0pt}
  \ifodd #1
   \setlength{\leftmargin}{\cslhangindent}
   \setlength{\itemindent}{-1\cslhangindent}
  \fi
  \setlength{\itemsep}{#2\baselineskip}}}
 {\end{list}}
\usepackage{calc}

\usepackage{amsmath}
\usepackage{amssymb}
\usepackage{amsthm}
\usepackage{mathtools}
\usepackage{epigraph}
\usepackage{thmtools}
\usepackage{thm-restate}
\usepackage{caption}
\declaretheorem[name=Proposition]{prop}

\declaretheorem[name=Assumption]{assn}
\makeatletter
\@ifpackageloaded{caption}{}{\usepackage{caption}}
\AtBeginDocument{%
\ifdefined\contentsname
  \renewcommand*\contentsname{Table of contents}
\else
  \newcommand\contentsname{Table of contents}
\fi
\ifdefined\listfigurename
  \renewcommand*\listfigurename{List of Figures}
\else
  \newcommand\listfigurename{List of Figures}
\fi
\ifdefined\listtablename
  \renewcommand*\listtablename{List of Tables}
\else
  \newcommand\listtablename{List of Tables}
\fi
\ifdefined\figurename
  \renewcommand*\figurename{Figure}
\else
  \newcommand\figurename{Figure}
\fi
\ifdefined\tablename
  \renewcommand*\tablename{Table}
\else
  \newcommand\tablename{Table}
\fi
}
\@ifpackageloaded{float}{}{\usepackage{float}}
\floatstyle{ruled}
\@ifundefined{c@chapter}{\newfloat{codelisting}{h}{lop}}{\newfloat{codelisting}{h}{lop}[chapter]}
\floatname{codelisting}{Listing}

\makeatother
\makeatletter
\@ifpackageloaded{caption}{}{\usepackage{caption}}
\@ifpackageloaded{subcaption}{}{\usepackage{subcaption}}
\makeatother
\usepackage{bookmark}
\IfFileExists{xurl.sty}{\usepackage{xurl}}{} % add URL line breaks if available
\hypersetup{
  pdftitle={Prices and Competition in Vertically Integrated Launch Markets},
  pdfauthor={Akhil Rao},
  pdfkeywords={vertical integration, launch
services, capacity-constrained pricing, foreclosure, space economics},
  colorlinks=true,
  linkcolor={blue},
  filecolor={Maroon},
  citecolor={Blue},
  urlcolor={Blue},
  pdfcreator={LaTeX via pandoc}}

\title{Prices and Competition in Vertically Integrated Launch Markets}
\author{Akhil Rao\footnote{Rational Futures, Washington, D.C. Email:
  akhil@rationalfutures.com. I am grateful to Tom Colvin, Moon Kim,
  Adrian Mangiuca, Jeff Feige and participants in the HBS/ESPI Space
  Economics Seminar Series for helpful discussions and feedback. All
  errors are my own.}}
\date{}
\begin{document}
\maketitle
\begin{abstract}
Over the last 15 years the number of U.S. orbital launches has grown by
roughly an order of magnitude. About three-quarters of those launches
were on SpaceX's Falcon 9 vehicle, and roughly three-fifths of those
Falcon launches deployed SpaceX's own Starlink constellation. A
back-of-envelope Wright's law calculation suggests this increase in
experience should have driven the Falcon 9's real launch cost down by
roughly 70\% over 2012--2026. Yet over the same period the advertised
price fell by less than 6\% in real terms. Why? I develop a simple model
of competition and vertical integration between launchers and
constellations. The launch market is Bertrand; the constellation
services market is Cournot; one launcher is integrated with its captive
constellation. Three results follow. First, the removal of double
marginalization raises the captive constellation's equilibrium size. If
the integrated launcher obtains cost reductions from this experience,
they are captured as capacity rent rather than passed through to
external buyers. Second, the integrated launcher prices launches to be
indifferent between serving internal and external demand, leaving more
residual demand for a competing launcher to monopolize and pushing the
equilibrium launch price up. Third, the same capacity rent that holds
the equilibrium launch price up can attract entry to the launch segment,
while the expansion of the captive constellation deters entry on the
constellation side.

\vspace{0.8em}\par\noindent\textbf{Keywords:} vertical integration, launch services, capacity-constrained pricing, foreclosure, space economics.\par\noindent\textbf{JEL classification:} L11, L22, L42, O33.
\end{abstract}

\setstretch{1.15}
\newpage

\setlength{\epigraphwidth}{0.4\textwidth}
\epigraph{...some industries, because they are integrated through to the consumer ... have not been faced with countervailing power.}{--- John Kenneth Galbraith, \\ \emph{American Capitalism: The Concept of Countervailing Power}}

\section{Introduction}\label{introduction}

Over the last 15 years the United States has gone from a handful of
orbital launches a year to well over a hundred. The majority of them
have been on a single vehicle, SpaceX's Falcon 9. The cost of aerospace
hardware often falls with cumulative production experience, the
learning-by-doing effect formalized as Wright's law
(\citeproc{ref-wright1936}{Wright 1936};
\citeproc{ref-nasaceh2004}{National Aeronautics and Space Administration
2004}; \citeproc{ref-terzinicoli2024}{Terzi and Nicoli 2024}). Applied
to the Falcon 9's cumulative flights at NASA's default eighty-five
percent learning rate, a Wright's law calculation predicts that the real
cost of a Falcon 9 launch should have fallen by roughly seventy percent
between 2012 and 2026. Under constant proportional markups, the price
should have fallen similarly. In fact, the advertised price fell by less
than six percent in real terms over the same window. Some customers even
saw real price increases across a broader set of launch vehicles over
this period (\citeproc{ref-kim2025counting}{Kim 2025}). Meanwhile,
roughly three-fifths of the Falcon 9's flights carry SpaceX's own
Starlink constellation, which is an order of magnitude larger than its
nearest U.S. competitor. The next-largest constellations own no launch
vehicles and buy their launches on the market, including from SpaceX.
Figure~\ref{fig-stylized-facts} illustrates these stylized facts.

\begin{figure}

\centering{

\includegraphics[width=1\linewidth,height=\textheight,keepaspectratio]{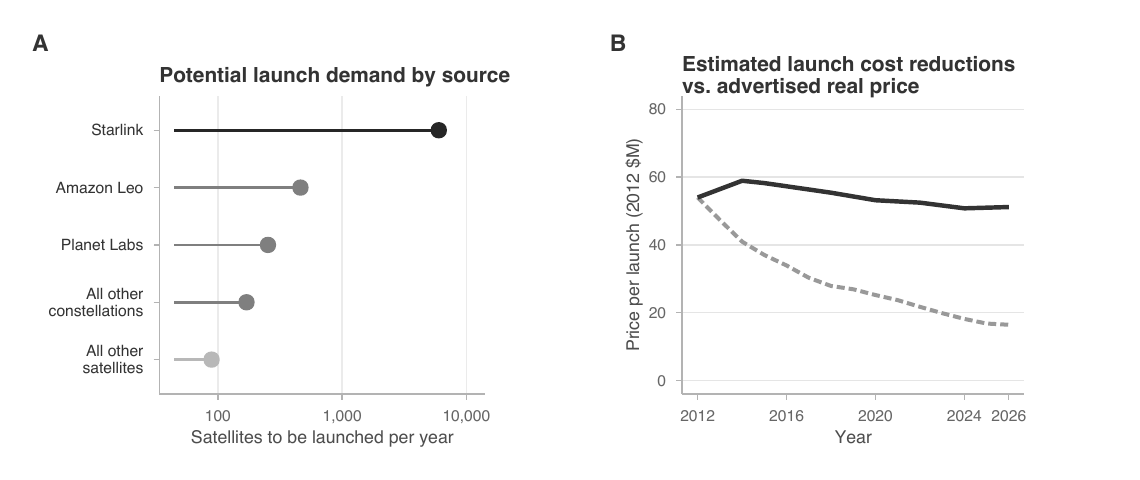}

}

\caption{\label{fig-stylized-facts}\textbf{Stylized facts about the U.S.
launch market.} (A) Potential U.S. launch demand by source. Of the named
constellations, only Starlink is vertically integrated. Estimates show
only systems that have received FCC Part 25 approval. (B) Falcon 9's
estimated launch price if Wright's law cost reductions were passed
through to maintain a constant margin (gray dashed line) and
inflation-adjusted advertised Falcon 9 prices (black solid line).
\emph{Sources:} (A) (\citeproc{ref-colvinkimrao2026}{Colvin, Kim, and
Rao 2026}), Scenario A annualized potential demand; (B) Falcon 9
advertised-price series, NNSI-deflated to 2012 dollars, against an
85\%-slope Wright's law prediction, both anchored at \$54M in 2012. The
source of each advertised-price observation is documented in the
appendix.}

\end{figure}%

I show that these facts are tightly connected. Launch is a common input
to every space market, so its market structure shapes competition
throughout the space economy. For a firm that launches its own
satellites, a launch sold to an outside buyer is a launch not flown for
its own constellation. Its reservation price for an external launch is
then the downstream profit it forgoes; a cost reduction widens that
margin. Vertical integration between a launcher and a constellation may
thus not pass cost reductions through, instead converting them into a
rent on scarce launch capacity.

I formalize this in a model of competition between two launchers and two
constellations, one launcher integrated with a captive constellation.
The model is static and highly stylized to illustrate the key economic
forces. Launchers set prices subject to a capacity constraint, and
constellations compete in quantities of satellites to provide satellite
services. Firms compete over external government and commercial payload
demand, which is small relative to the integrated firm's own launch
demand but large enough to sustain more than one launcher in equilibrium
(consistent with U.S. launch demand estimates in Colvin, Kim, and Rao
(\citeproc{ref-colvinkimrao2026}{2026})). Launch capacities, costs,
non-constellation demand, and the integration structure are primitives.
I use the model to compare the Nash equilibria of an integrated and a
non-integrated industry.

The core mechanism is an opportunity cost: the value of a marginal
launch to the integrated firm is the downstream profit it could earn
from launching satellites to its own constellation. This shadow price
sets a floor below the external price, and creates rents for the
remaining launcher. Three results follow. First, even when the two
launchers have identical costs, vertical integration enables the captive
constellation to grow larger than the independent constellation. The
integrated launcher may thus accumulate more experience than the
non-integrated launcher, and if learning-by-doing applies, obtain lower
marginal costs. This cost advantage is captured as a capacity rent
rather than passed through. Second, because the integrated launcher
internalizes the capacity rent and reduces its launch supply to the
external market, the non-integrated launcher is able to charge higher
prices. The equilibrium price thus reflects the capacity rent even when
the marginal launcher for external demand is not integrated. Third, this
rent affects entry asymmetrically: it can attract entry into the launch
segment while deterring entry into the constellation segment.

The paper sits at the intersection of three literatures. The first is
vertical foreclosure: when the integrated launcher withholds scarce
launch capacity from the external market, the price independent
constellations pay rises, echoing results in Salinger
(\citeproc{ref-salinger1988}{1988}), Ordover, Saloner, and Salop
(\citeproc{ref-ordoversalonersalop1990}{1990}), Hart and Tirole
(\citeproc{ref-harttirole1990}{1990}), and Rey and Tirole
(\citeproc{ref-reytirole2007primer}{2007}). Unlike that literature, cost
reductions in this model are retained as rent on scarce capacity. The
second is pricing of space launch services. Terzi and Nicoli
(\citeproc{ref-terzinicoli2024}{2024}) estimate how Wright's law may
materialize in the launch industry, and Su, Yang, and Sweeting
(\citeproc{ref-suyangsweeting2026}{2026}) show how mergers like the
United Launch Alliance joint venture may increase buyer (government) and
total social surplus in the presence of learning-by-doing. This paper
assumes learning effects as an exogenous static cost differential. Rao
and Colvin (\citeproc{ref-raocolvin2025}{2025}) assess how SpaceX's
pricing of Starship may depend on Starlink's profitability when launches
are scarce. This model recovers that argument as a property of the joint
first-order condition, with a modification for the slope of residual
external demand. The third is the economics of innovation and entry
(e.g., Aghion and Howitt (\citeproc{ref-aghionhowitt1992}{1992})). The
result in the constellation segment echoes the efficiency effect of
Gilbert and Newbery (\citeproc{ref-gilbertnewbery1982}{1982}), under
which a dominant incumbent preempts and maintains market power, while
the result in the launch segment is reminiscent of the replacement
effect of Arrow (\citeproc{ref-arrow1962}{1962}) and Reinganum
(\citeproc{ref-reinganum1983}{1983}), under which a capable entrant can
still displace an incumbent. Doyle (\citeproc{ref-doyle2026}{2026})
studies a related model in which an oligopoly of constellations
endogenously reduces innovation without a launch market; similarly,
Guyot, Rao, and Rouillon (\citeproc{ref-guyotetal2023}{2023}) study a
model in which a constellation duopoly reduces economic welfare below
the social optimum. Neither paper includes the launch segment, which is
a key aspect of the model developed here.

The rest of the paper proceeds as follows. Section 2 sets out the model
and derives how vertical integration reshapes the equilibrium
constellation sizes and the external launch price. Section 3 turns to
entry, showing that the same capacity rent attracts entry in the launch
segment while deterring it in the constellation segment. Section 4
discusses what the results may imply for competition and creative
destruction in launch, and concludes. All proofs are in the appendix.

\section{Vertical integration and
constellations}\label{vertical-integration-and-constellations}

In this section I show that vertical integration between a launcher and
a constellation expands the size of the captive constellation compared
to its non-integrated rival. Learning-by-doing may thus generate
differential launch cost reductions for the integrated firm, which the
integrated firm captures as rents. I consider two market configurations:

\begin{itemize}
\item Configuration $NI$, in which there is no vertical integration. Both launchers compete to serve two independent constellations and external demand, and both constellations purchase launches at the market price.
\item Configuration $VI$, in which one launcher and constellation are integrated and their choice variables are jointly optimized, while the other constellation purchases launches at the market price.
\end{itemize}

The configuration assumed is denoted by superscripts over variables,
e.g., \(N_i^{NI}\) is the size of constellation \(i\) under
configuration \(NI\).

\subsection{The setting}\label{the-setting}

There are two launchers, \(S\) and \(B\), with marginal launch costs
\(l_B \geq l_S > 0\). Under learning-by-doing, \(l_i\) is decreasing in
the cumulative number of launches \(i\) conducts. For simplicity, I
ignore the time dimension throughout the analysis and consider cases
where \(l_B = l_S\) and \(l_B > l_S\).

Launcher \(B\) has unlimited capacity, while launcher \(S\) has capacity
\(k_S > 0\). Launch capacity here may be interpreted as rocket
production capacity as in Triezenberg et al.
(\citeproc{ref-triezenberg2024}{2024}), or as launch permits as in
Colvin, Kim, and Rao (\citeproc{ref-colvinkimrao2026}{2026}). What
matters is not the exact nature of the constraint but two facts: no firm
can launch infinitely many flights, and no single firm has captured the
whole launch market (including launch demand of other
constellations).\footnote{If the dominant firm had unlimited capacity
  and lower costs, the standard Bertrand outcome degenerates to the
  dominant firm capturing the entire launch market. Given that this has
  not been observed over the past 15 years, I make assumptions to focus
  on the empirically relevant situation. While it is more realistic for
  both launchers to have limited capacity, Assumption
  \ref{assn:demand-small} ensures launcher \(B\) launches no more than
  launcher \(S\) in equilibrium, so \(B\)'s capacity does not bind. This
  roughly matches reality: SpaceX appears to have more launch capacity
  than its next nearest launch competitors, but its competitors have
  still captured a positive number of launches every year. The central
  economic force, the capacity rent that induces \(S\) to prioritize its
  launch supply for itself, does not depend on \(B\)'s capacity being
  unlimited.} There are two constellations, \(S\) and \(K\), with
\(N_S\) and \(N_K\) satellites respectively. Constellation \(S\) is
captive to launcher \(S\) under vertical integration, while
constellation \(K\) is always independent. Satellites in each
constellation cost \(c\) to build and operate and require a single
launch.

In addition to the two constellations, there is external demand for
launches \begin{equation}
X(p) = a - bp,
\end{equation}

where \(a, b > 0\). The external demand represents all other space
systems, e.g., government missions and constellations in other services
markets. Launcher \(S\) sells \(m_S\) launches to external customers
including the other constellation, so its total launches are
\(k_S = N_S + m_S\) and launcher \(B\)'s are \(Q_B = N_K + X(p) - m_S\).
Launchers compete for launch demand by setting prices, resulting in
equilibrium launch price \(p_{eqm}\).

I place two assumptions on demand primitives: first that external
non-constellation demand is small relative to the integrated launcher's
capacity, and second that launcher \(S\) does not claim the entire
launch market. The first is consistent with Colvin, Kim, and Rao
(\citeproc{ref-colvinkimrao2026}{2026}), whose estimates place U.S.
government and commercial launch demand from non-captive constellations
below available launch capacity, with capacity pressure concentrated in
the integrated constellation operators' own launch demand. The second is
consistent with the observed fact that over the past 15 years, no single
U.S. launcher has captured all U.S. launch demand.

\begin{assn}[External demand is small relative to capacity]
\label{assn:demand-small}
Background launch demand is small relative to launcher $S$'s capacity, $a \leq \bar a$ with $\bar a = \tfrac{6 k_S(1 + 3 b\gamma) + b\,(4c + 5 l_B + 6 b\gamma\, l_B - 4\theta)}{1 + 6 b\gamma}$, so that in equilibrium launcher $B$ launches no more than launcher $S$.
\end{assn}

\begin{assn}[Launcher $S$ does not take the whole market]
\label{assn:capacity-binds}
Launcher $S$'s capacity is large enough that it sells some launches externally ($m_S >0$), but small enough that its capacity constraint binds without capturing the whole market ($m_S < N_K + X(p)$). Formally, $\underline{k} < k_S < \bar{k}$, with $\bar{k} = \tfrac{a(1+6b\gamma) - b\,(4c + 5l_B + 6 b\gamma\, l_B - 4\theta)}{6 b\gamma}$ and $\underline{k} = \tfrac{2 b\theta - a - b\,(2c + l_B)}{6 b\gamma}$.
\end{assn}

Intuitively, Assumptions \ref{assn:demand-small} and
\ref{assn:capacity-binds} ensure that \(S\)'s launch capacity is large
enough to sell externally but small enough to bind. Operating
capacity-constrained, rather than building out enough capacity to serve
the market with slack, is a standard outcome from an optimal capacity
choice problem when capacity is costly.\footnote{In the peak-load
  capacity pricing literature, price exceeds marginal operating cost
  when capacity binds (\citeproc{ref-williamson1966}{Williamson 1966}).
  Optimal capacity is set by equating the marginal cost of capacity to
  the expected marginal cost of unserved demand, leaving demand above
  capacity with positive probability (\citeproc{ref-chao1983}{Chao
  1983}). The appendix shows a version of this result applies here: when
  the integrated firm builds its launch capacity at a convex cost before
  the market game, the marginal cost of launch capacity is equated with
  the capacity rent and launch capacity binds in equilibrium.}

Constellations compete to provide services to end users by choosing the
number of satellites they deploy.\footnote{This type of Cournot
  structure can obtain from a game in which service quality is
  increasing in the number of satellites, as is the case for LEO
  broadband services. For example, Guyot, Rao, and Rouillon
  (\citeproc{ref-guyotetal2023}{2023}) and Doyle
  (\citeproc{ref-doyle2026}{2026}) study games between
  telecommunications constellation operators with this feature.} They
are faced with linear inverse demand \begin{equation}
P_s(N_S + N_K) = \theta - \gamma (N_S + N_K),
\end{equation} where \(\theta, \gamma >0\). Constellation revenues are
given by \(R(N_i) = P_s(N_S + N_K) N_i\). When a constellation buys
launches at the market price \(p_{eqm}\), its effective per-satellite
cost is \(c + p_{eqm}\); when constellation \(S\) is vertically
integrated with launcher \(S\), the joint firm's effective per-satellite
cost is \(c + l_S\). Each constellation chooses its size \(N_i\) to
maximize gross services revenues net of the costs of operating and
launching its satellites.

\textbf{Equilibrium concept.} Launchers post launch prices,
constellations choose sizes and the launch market clears under efficient
(surplus-maximizing) rationing.\footnote{Most of what follows does not
  turn on the specifics of the rationing rule. The captive
  constellation's size advantage (Proposition
  \ref{prop:captive-expansion}), the opportunity-cost lemma (Lemma
  \ref{lemm:opportunity-cost}), and constellation-entry deterrence
  (Proposition \ref{prop:vertical-integration-constellation-entry}) need
  only that launcher \(B\) faces a downward-sloping residual demand that
  the integrated firm's withdrawal of external supply shifts outward.
  The signs and magnitudes of the price wedge in Proposition
  \ref{prop:vertical-integration-launch-price} and the attraction of a
  launch entrant in Proposition
  \ref{prop:vertical-integration-launch-entry} depend on the relative
  degree to which rationing is efficient vs.~proportional.} The
integrated firm chooses its captive size \(N_S\) and external launch
supply \(m_S\) jointly to maximize combined profit. An equilibrium is a
launch price (\(p_{eqm}\)) and a profile of launcher supplies and
constellation sizes (\(m_S, Q_B, N_S, N_K\)) at which no firm can
profitably deviate holding the others' choices fixed.

Both launchers and the independent constellation internalize their
effects on the equilibrium launch price. \(S\) internalizes this through
its external supply \(m_S\), \(B\) internalizes this through the
dependence of its residual-monopoly price on the external supply \(S\)
offers, and the independent constellation \(K\) internalizes this
through its launch demand \(N_K\).\footnote{Throughout I take the
  independent constellation \(K\) to internalize its effect on the
  launch price. Whether \(K\) internalizes its price impact changes the
  magnitude of the gap, not its sign; the expression is algebraically
  cleaner when \(K\) internalizes its price effect. In reality, a large
  non-integrated operator like Amazon Leo may purchase launches in large
  enough tranches to materially affect market supply
  (\citeproc{ref-foust2022kuiper}{Foust 2022}).}

\subsection{Equilibrium constellation sizes and vertical
integration}\label{equilibrium-constellation-sizes-and-vertical-integration}

The following proposition establishes that, given Assumptions
\ref{assn:demand-small} and \ref{assn:capacity-binds}, the captive
constellation \(S\) is strictly larger than the independent
constellation \(K\) even when launch costs are identical between \(S\)
and \(B\).

\begin{restatable}[Captive expansion under vertical integration]{prop}{propConstellationSize}
\label{prop:captive-expansion}
Given Assumptions \ref{assn:demand-small} and \ref{assn:capacity-binds}, the captive constellation is strictly larger than the independent constellation, with gap
\begin{equation}
N^{VI}_S - N^{VI}_K = \dfrac{k_S}{1 + 2 b \gamma} > 0,
\end{equation}
even at symmetric launch costs $l_S = l_B$.
\end{restatable}

The key condition of the proof derives from the joint first-order
condition \(S\) optimizes, taking \(B\)'s behavior into account:
\begin{equation}
\label{eqn:opp-cost-engine}
\underbrace{R'(N_S) - c}_{\text{captive margin}} = \underbrace{p_{eqm} - \tfrac{m_S}{2b}}_{\text{external margin}}.
\end{equation}

Intuitively, \(S\) pays \(l_S\) on all \(k_S\) launches whatever the
internal/external split, so \(l_S k_S\) is a fixed cost that drops out
of the marginal allocation. As a result, the optimal split balances the
margin earned on captive satellites against the margin earned from
external sales. Like the representative agent in Samuelson
(\citeproc{ref-samuelson1956}{1956}), the integrated firm jointly
optimizes its capacity allocation across multiple uses. Marginal revenue
is then equated across uses of the scarce launch resource, both within
the constellation segment where \(S\) and \(K\) compete and across other
industries captured in \(X(p)\). Rao and Colvin
(\citeproc{ref-raocolvin2025}{2025}) make this argument under
price-taking behavior. Here \(S\) internalizes its impact on the market
price, so the external margin it perceives is the equilibrium price
\(p_{eqm}\) minus \(m_S/2b\). The latter term reflects that \(S\)
internalizes that selling one more launch externally lowers the price it
earns on all \(m_S\) of its external launches. Because it buys its
launches on the market, the independent constellation pays launcher
\(B\)'s markup on every satellite, whereas the integrated firm supplies
its own captive constellation at its internal cost and escapes the
markup (i.e., removes double marginalization). The captive therefore
faces a lower effective per-satellite cost and optimally deploys more,
even when the two launchers' costs are identical.
Figure~\ref{fig-vi-mechanism} illustrates the condition (panel A) and
the effect of vertical integration on the equilibrium constellation
sizes (panel B).

\begin{figure}

\centering{

\includegraphics[width=1\linewidth,height=\textheight,keepaspectratio]{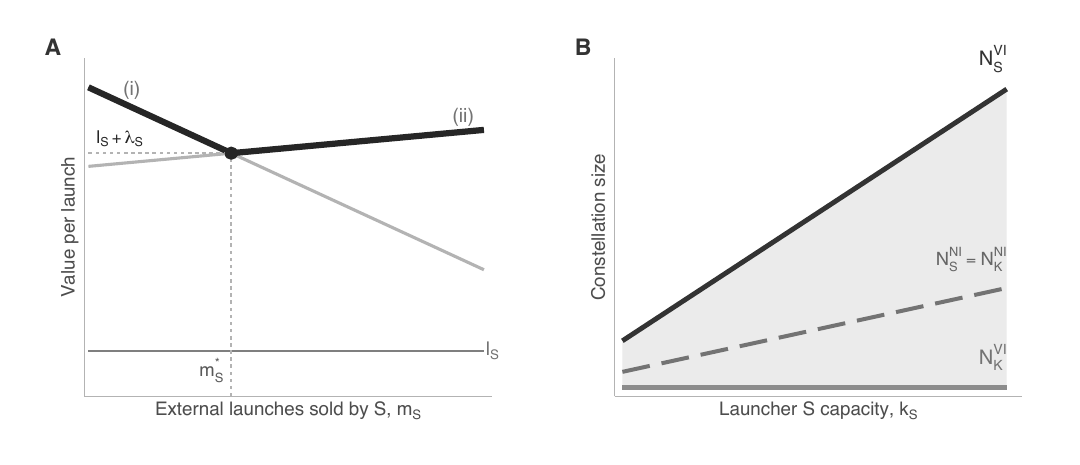}

}

\caption{\label{fig-vi-mechanism}\textbf{The integrated firm's capacity
allocation choice and the captive's expansion.} (A) The value of a
launch to the integrated firm is the upper envelope of its two uses as a
function of external sales \(m_S\): (i) shows the value of an external
sale, which is falling in \(m_S\), and (ii) shows the value of using
launch capacity for the captive constellation, which is decreasing in
\(N_S\) (so increasing in \(m_S\)). The firm sells externally up to
\(m_S^{*}\), where the two are equal and define the reservation price
\(l_S+\lambda_S\). (B) Captive expansion. Equilibrium constellation
sizes against \(S\)'s launch capacity \(k_S\). The captive \(N_S^{VI}\)
grows with capacity while the independent \(N_K^{VI}\) stays flat. The
shaded gap equals \(k_S/(1+2b\gamma)\), which is strictly positive and
independent of \(l_S\). The dashed line is the symmetric no-integration
benchmark where \(N_S^{NI}=N_K^{NI}\). \emph{Illustrative parameters:}
\(\theta=8\), \(a=11\), \(\gamma=0.1\), \(c=0.5\), \(l_S=1\), \(l_B=2\),
\(b=1\), \(k_S=12\).}

\end{figure}%

The form of Proposition \ref{prop:captive-expansion} shows that, when
launcher \(S\) also sells externally, an increase in its launch capacity
is divided between the captive constellation and external launch sales,
with the captive taking the share given in Corollary
\ref{coro:captive-constellation-captures-capacity}.

\begin{restatable}[The captive's share of capacity growth]{coro}{coroAllocation}
\label{coro:captive-constellation-captures-capacity}
Given Proposition \ref{prop:captive-expansion}, whenever launcher $S$ also sells externally ($m_S > 0$), the captive constellation captures the share
\begin{equation}
\frac{\partial N_S}{\partial k_S} = \frac{1}{1 + 2 b \gamma}
\end{equation}
of any increase in launcher $S$'s launch capacity, the remaining $2 b \gamma/(1 + 2 b \gamma)$ accruing to external launch sales.
\end{restatable}

The denominator \(1+2b\gamma\) can be interpreted in terms of the
sources of value \(S\) must balance. A unit of launch capacity is worth
deploying internally up to the point where a marginal captive satellite
is worth no more than a launch sold outside. What governs the division
is how fast each use loses value as it takes up more launches. Adding a
satellite to the captive constellation's allocation lowers the price the
whole constellation earns at rate \(2 \gamma\), while selling a launch
externally lowers the value of an external sale at rate \(1/b\). The
ratio of the two erosion rates is \(2b\gamma\). When the services market
is deep and outside launch demand is unresponsive (\(\gamma\) and \(b\)
are both small), internal use of launch capacity barely cannibalizes
itself and the captive constellation absorbs almost the whole capacity
increase (\(1+2b\gamma \approx 1\)). When either market is thin
(\(\gamma\) or \(b\) are large), the marginal launch may be better sold
than flown.

\subsection{Equilibrium launch price and vertical
integration}\label{equilibrium-launch-price-and-vertical-integration}

Having assessed how the constellation sizes respond to vertical
integration, next I turn to the response of the launch price. As
described in Rao and Colvin (\citeproc{ref-raocolvin2025}{2025}),
launcher \(S\)'s opportunity cost of an external launch is the internal
margin it forgoes by selling a launch rather than flying its own
satellite. This defines the floor price below which it will not sell.
Under Assumption \ref{assn:capacity-binds} launcher \(S\) serves
(captive and other) launch demand up to \(k_S\), and launcher \(B\) sets
the price on the residual demand \(S\) cannot serve. Launcher \(B\)'s
market power then raises the equilibrium launch price above \(l_B\),
launcher \(B\)'s marginal cost. Lemma \ref{lemm:opportunity-cost} and
Proposition \ref{prop:vertical-integration-launch-price} make these
intuitions precise.

\begin{restatable}[Opportunity cost of an external launch]{lemm}{lemmOppCost}
\label{lemm:opportunity-cost}
Given Assumptions \ref{assn:demand-small} and \ref{assn:capacity-binds}, the shadow price $\lambda_S$ of launcher $S$'s capacity constraint equals its reservation price for an external launch, net of the launch cost:
\begin{equation}
\lambda_S = (p_{eqm} - l_S) - \frac{m_S}{2b}.
\end{equation}
\end{restatable}

I refer to \(\lambda_S\) (the shadow price of launcher \(S\)'s capacity
constraint) as the capacity rent. It is the premium a scarce launch slot
earns above its marginal cost \(l_S\), reflecting the internal margin
\(S\) forgoes by selling a launch rather than flying its own satellite.

\begin{restatable}[Vertical integration and the equilibrium launch price]{prop}{propLaunchPrice}
\label{prop:vertical-integration-launch-price}
Given Assumptions \ref{assn:demand-small}, \ref{assn:capacity-binds}, and efficient rationing, the equilibrium launch price is launcher $B$'s residual-monopoly price on the demand $S$ leaves unserved, which decomposes into $S$'s opportunity cost of an external launch plus $B$'s residual markup,
\begin{equation}
p^{VI}_{eqm} = \frac{a + N_K - m_S + b l_B}{2b} = \big(R'(N_S) - c\big) + \frac{m_S}{2b},
\end{equation}
and lies strictly above the non-integrated price,
\begin{equation}
p^{VI}_{eqm} - p^{NI}_{eqm} = \frac{k_S}{6 b (1 + 2 b\gamma)} > 0.
\end{equation}
Changes to the integrated firm's launch cost are not passed through:
\begin{equation}
\frac{\partial p^{VI}_{eqm}}{\partial l_S} = 0.
\end{equation}
\end{restatable}

Lemma \ref{lemm:opportunity-cost} shows that the integrated firm's
shadow price of an external launch is the gross revenue from an external
launch minus a wedge \(m_S/2b\), and Proposition
\ref{prop:vertical-integration-launch-price} shows that the equilibrium
launch price is the integrated firm's marginal revenue from another
captive satellite plus the same wedge. Under \(VI\), this wedge leads
\(S\) to supply less externally (\(m_S\) is smaller) to serve its
now-larger captive constellation, which in turn enlarges the residual
demand \(B\) prices as a monopolist. Figure~\ref{fig-prop2} illustrates
the shift in residual demand faced by \(B\) under vertical integration.

\begin{figure}

\centering{

\includegraphics[width=0.9\linewidth,height=\textheight,keepaspectratio]{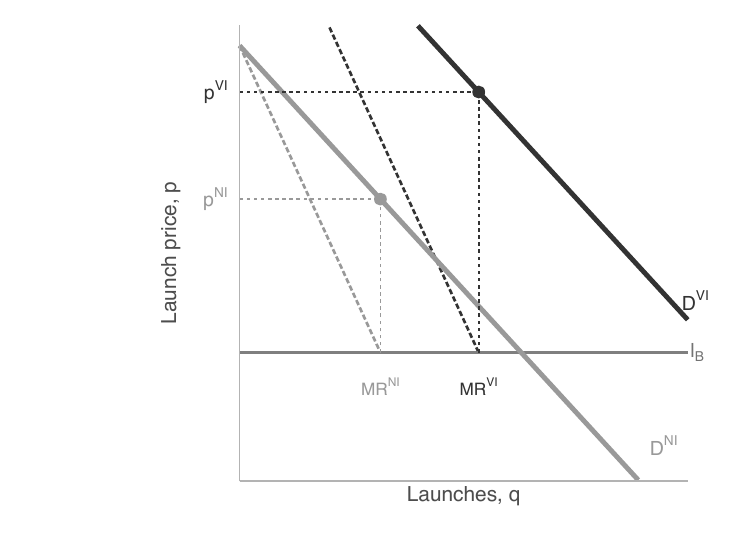}

}

\caption{\label{fig-prop2}\textbf{The launch price under vertical
integration.} Launcher \(B\) is the residual monopolist on the launch
demand \(S\) does not serve, pricing where its marginal revenue equals
its marginal cost \(l_B\). \(D^{NI}\) and \(D^{VI}\) are \(B\)'s
residual demand without and with integration. Under integration the
captive's launches go off-market and \(S\) withholds external supply,
shifting \(B\)'s residual demand outward, so the monopoly price rises
from \(p^{NI}\) to \(p^{VI}\) by \(k_S/(6b(1+2b\gamma))\). The
integrated launcher's own cost \(l_S\) does not enter.
\emph{Illustrative parameters:} \(\theta=8\), \(a=11\), \(\gamma=0.1\),
\(c=0.5\), \(l_S=1\), \(l_B=2\), \(b=1\), \(k_S=12\).}

\end{figure}%

Vertical integration also leaves the independent constellation \(K\)
smaller, through two compounding forces: the captive is larger
(Proposition \ref{prop:captive-expansion}), which crowds \(K\) through
Cournot substitution, and the launch price \(K\) pays is higher
(Proposition \ref{prop:vertical-integration-launch-price}). This can be
seen in Figure~\ref{fig-vi-mechanism}.

\begin{restatable}{coro}{coroSize}
\label{coro:independent-constellation-size-effect}
Given Propositions \ref{prop:captive-expansion} and \ref{prop:vertical-integration-launch-price}, the independent constellation is smaller under configuration $VI$ than under configuration $NI$, $N_K^{VI} < N_K^{NI}$.
\end{restatable}

Formally, the effective per-satellite cost \(K\) pays is \(c + p_{eqm}\)
in both configurations. Proposition
\ref{prop:vertical-integration-launch-price} shows
\(p^{VI}_{eqm} > p^{NI}_{eqm}\), so \(K\)'s launch cost is higher under
integration. Together with the larger captive crowding the services
market, both forces leave \(K\) smaller.

It also follows that the integrated launcher's cumulative launch
experience will exceed the non-integrated launcher's. If
learning-by-doing applies, it will have a lower marginal launch cost
than the non-integrated launcher.

\begin{restatable}{coro}{coroLearning}
\label{coro:learning-effects}
Let each launcher's marginal cost be a decreasing function of its equilibrium launch volume, $l_i = l(Q_i)$ with $l'(\cdot) \leq 0$ and strict inequality over a non-empty interval. Under Assumption \ref{assn:demand-small}, $Q_S \geq Q_B$, so $l_S \leq l_B$.
\end{restatable}

Proposition \ref{prop:captive-expansion} shows that the captive expands
even at equal launch costs, so launcher \(S\)'s larger volume (and the
cost edge it earns) does not presuppose the ordering it produces. For
the remainder of the analysis I take \(l_S < l_B\) under configuration
\(VI\) as a maintained primitive, with learning as its empirical
justification.\footnote{As the integrated constellation scales up,
  scarcity of productive inputs to \(S\)'s launch operations may bind
  and push launch costs back up. I am neglecting these effects since
  learning-by-doing operates over time, and I am simplifying the model
  to a single period for expositional clarity. A more detailed model
  with learning over time and U-shaped marginal costs within a period
  may reveal richer dynamics, though it will not change the basic point:
  vertical integration can give an asymmetric cost advantage to the
  integrated firm's launch operations.} Note that when Corollary
\ref{coro:independent-constellation-size-effect} holds, launcher \(B\)
may have less cumulative launch experience in the vertically integrated
case than in the non-integrated case. Given learning-by-doing, its
marginal cost \(l_B\) (and hence the launch price it sets) may be raised
higher when \(S\) vertically integrates.

\section{Vertical integration and
entry}\label{vertical-integration-and-entry}

Having examined how vertical integration reshapes equilibrium
constellation sizes and the launch price, I turn to entry. The forces
identified in Propositions \ref{prop:captive-expansion} and
\ref{prop:vertical-integration-launch-price} push in opposing directions
in the launch and constellation segments and generate asymmetric
foreclosure. Rents in the launch segment attract new launchers, which
may not reduce the price of launch if the entrants cannot serve all the
residual demand. The higher launch price, combined with the integrated
firm's dominance in the constellation segment, makes entry less
attractive to new constellations.

\subsection{The setting}\label{the-setting-1}

A new launcher \(E\) with marginal cost \(l_E \in (l_S, l_B)\) and
launch capacity \(k_E \in (0,\infty)\) may enter the launch segment by
paying fixed cost \(F_l\). Similarly, a new constellation \(A\) with
marginal per-satellite production and operation cost \(c_A \leq c\) may
enter the constellation segment by paying fixed cost \(F_c\). Formally,
let \(Q_E\) be the total launch demand launcher \(E\) may capture in
equilibrium. Launcher \(E\) will enter if and only if \begin{equation}
(p_{eqm} - l_E)Q_E > F_l.
\end{equation}

Similarly, let \(N_A\) be the equilibrium size of constellation \(A\).
Constellation \(A\) will enter if and only if \begin{equation}
R(N_A) - (c_A + p_{eqm}) N_A > F_c.
\end{equation}

\subsection{Entry in the launch and constellation
segments}\label{entry-in-the-launch-and-constellation-segments}

I first consider entry into the launch segment. Whether vertical
integration helps or hurts a launch entrant turns on the entrant's own
capacity: it raises the maximum sustainable fixed cost for a
capacity-limited entrant, but is neutral to one large enough to compete
the launch price down to \(l_B\).

\begin{restatable}[Vertical integration and launch segment entry]{prop}{propLaunchEntry}
\label{prop:vertical-integration-launch-entry}
Let $\hat{p}_{eqm}$ denote the post-entry launch price and $Q_E$ the volume $E$ captures, so the entrant's maximum sustainable fixed cost is $\tilde{F}_l = (\hat{p}_{eqm} - l_E)\,Q_E$. Its response to vertical integration depends on the entrant's capacity $k_E$.
\begin{enumerate}
\item \emph{Limited capacity $k_E$} ($E$ sells $Q_E = k_E$ inframarginally as a price-taker while $B$ prices the residual). Under efficient rationing, vertical integration \emph{raises} the maximum sustainable fixed cost,
\begin{equation}
\tilde{F}_l^{VI} - \tilde{F}_l^{NI} = \frac{k_E\, k_S}{6 b (1 + 2 b\gamma)} > 0,
\end{equation}
so it attracts launch entry.
\item \emph{Unlimited capacity.} Since $l_E < l_B$, $E$ undercuts $B$ and competes the launch price to $l_B$ in both configurations, so the maximum sustainable fixed costs coincide, $\tilde{F}_l^{VI} = \tilde{F}_l^{NI}$: integration is neutral to launch entry.
\end{enumerate}
\end{restatable}

The intuition is simple: the capacity rent that holds the launch price
above cost (Proposition \ref{prop:vertical-integration-launch-price}) is
a rent a capacity-limited entrant can capture, so integration makes the
launch market more attractive when the entrant has limited capacity and
cannot fully erode the rent.\footnote{Note that in the unlimited
  capacity case, the new entrant effectively displaces the old
  less-efficient firm. As long as \(B\) may still compete, the launch
  market price is pinned at \(l_B\); if \(B\) exits, \(E\) may then take
  its place as a monopolist on residual demand.} On the other hand, in
the constellation segment the new entrant faces lower rents under
vertical integration than when \(S\) is not integrated. Proposition
\ref{prop:vertical-integration-constellation-entry} shows this formally.

\begin{restatable}[Vertical integration and constellation segment entry]{prop}{propConstellationEntry}
\label{prop:vertical-integration-constellation-entry}
Let the maximum sustainable fixed cost in the constellation segment be $\tilde{F}_c = R(N_A) - (c_A + p_{eqm})N_A$, the entrant's post-entry operating profit. Given Assumptions \ref{assn:demand-small} and \ref{assn:capacity-binds}, vertical integration lowers it whenever the entrant is active in both configurations,
\begin{equation}
\tilde{F}_c^{VI} - \tilde{F}_c^{NI} = -\frac{k_S\big(N_A^{VI} + N_A^{NI}\big)}{8b} = -\frac{k_S\big(2b(2\theta + 4c - 6 c_A - l_B) + k_S - 2a\big)}{32\,b\,(1 + 2 b\gamma)} < 0:
\end{equation}
vertical integration \emph{deters} constellation entry. The entrant is active under $VI$ only if it is efficient enough,
\begin{equation}
c_A < \frac{b(2\theta + 4c - l_B) - a}{6b},
\end{equation}
a threshold the non-integrated configuration reduces by $k_S/6b$, so that some entrants which may enter without integration are foreclosed by it.
\end{restatable}

Intuitively, the enlarged captive crowds the services market: by Cournot
substitution a larger \(S\) shrinks any entrant \(A\), and the rent
\(A\) could earn falls with it. To match the profit it would earn
without \(S\) integrated, \(A\) must be efficient enough to offset its
size disadvantage against the captive constellation \(S\). Potential
entrants that would clear that bar without integration are therefore
foreclosed by it. Learning-by-doing may compound the disadvantage over
time, as the larger captive's launches feed cost reductions the entrant
cannot match. Figure~\ref{fig-props34-entry} shows both effects.

\begin{figure}

\centering{

\includegraphics[width=1\linewidth,height=\textheight,keepaspectratio]{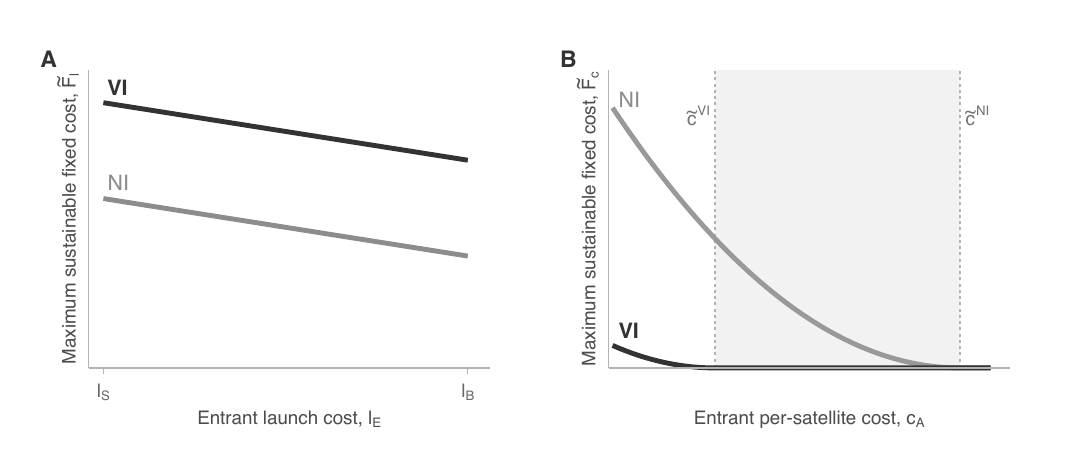}

}

\caption{\label{fig-props34-entry}\textbf{The asymmetric entry effects
of vertical integration.} (A) Launch-segment entry: the maximum
sustainable fixed cost \(\tilde{F}_l\) a capacity-limited entrant can
bear, against its launch cost \(l_E\). Vertical integration raises it
(\(VI\) above \(NI\)), so integration attracts launch entry. (B)
Constellation-segment entry: the maximum sustainable fixed cost
\(\tilde{F}_c\) of a fresh constellation entrant, against its
per-satellite cost \(c_A\). Vertical integration lowers it (\(VI\) below
\(NI\)) and forecloses potential entrants between the \(VI\) and \(NI\)
activity thresholds, \(\tilde{c}^{VI} \le c_A < \tilde{c}^{NI}\)
(shaded), so integration deters constellation entry. \emph{Illustrative
parameters:} \(\theta=8\), \(a=11\), \(\gamma=0.1\), \(c=0.5\),
\(l_S=1\), \(l_B=2\), \(b=1\), \(k_S=12\), \(k_E=2\).}

\end{figure}%

\section{Discussion and conclusion}\label{discussion-and-conclusion}

Vertical integration between a launcher and a constellation may lower
the cost of launch while increasing its market price. The removal of
double marginalization enlarges the captive constellation (Proposition
\ref{prop:captive-expansion}), which may give the integrated launcher
greater cumulative flight experience than competing launchers. Whether
this experience translates to engineering cost reductions or not,
profitability in the constellation segment creates an opportunity cost
to the integrated launcher from allocating scarce launch capacity to
external customers. This allows non-integrated launchers to raise prices
as well (Proposition \ref{prop:vertical-integration-launch-price}). The
opportunity cost thus creates a rent on the integrated launcher's
capacity. Any cost reductions the integrated launcher achieves, such as
through learning-by-doing, are absorbed into the rent rather than passed
through to other launch customers
(\(\partial p^{VI}/\partial l_S = 0\)). That same rent may invite entry
into the launch segment (Proposition
\ref{prop:vertical-integration-launch-entry}), while foreclosing entry
into the constellation segment (Proposition
\ref{prop:vertical-integration-constellation-entry}).

These results are consistent with the stylized facts that motivate the
paper: the vertically integrated constellation is larger than the
non-integrated constellation, and the launch price does not fall despite
the presence of a competitor even if the integrated launcher achieves
cost reductions. As Weinzierl (\citeproc{ref-weinzierl2018jep}{2018})
notes, complementarities in the space sector may push toward integration
and concentration. As a common input to space activities, launch prices
tie the cost of access across markets to what a vertically integrated
launcher may earn from its most profitable constellation segment. High
profitability in one downstream segment may therefore increase launch
prices across satellite markets. In addition to persistent dominance in
profitable markets, commonalities in satellite design across segments
(e.g., consumer broadband and backhaul) may enable a vertically
integrated firm to use its capacity rent to finance horizontal
downstream expansion. If this limits the demand other launchers can
capture, it may create a recurring pattern of launch entrants attracted
by rents but unable to achieve scale due to scarce demand.

The model is deliberately spare. It treats launch as a homogeneous good,
ignoring relevant differentiation along axes like mass-to-orbit and
reachable orbits as modeled in Triezenberg et al.
(\citeproc{ref-triezenberg2024}{2024}) and Colvin, Kim, and Rao
(\citeproc{ref-colvinkimrao2026}{2026}). Many assumptions, such as
linear demand functions and exogenous capacity, are chosen for simple
closed forms. Learning-by-doing in space launch is modeled as a fixed
cost advantage for the firm that integrates. Similar learning effects in
satellite bus production may deepen the foreclosure by enabling
horizontal downstream expansion but are not modeled. These
simplifications allow the model to isolate key economic forces, and open
new ones. For example, when does the integrated firm's cost advantage
compound through sustained learning and entrench the rent? Will vertical
integration by one firm incent integration by others, and if so what
will happen to launch and downstream market prices? Under what
conditions might vertical and horizontal integration across
constellation markets generate enough competition to reduce launch
prices? How might orbital debris or space traffic management constraints
affect the conduct and profitability of integrated and non-integrated
constellations? Answers to these and other economic questions about
proliferated space systems may depend on the capacity rents accruing to
integrated launchers.

\newpage

\section*{References}\label{references}
\addcontentsline{toc}{section}{References}

\phantomsection\label{refs}

\newpage

\section*{Appendix A: Proofs}\label{appendix-a-proofs}
\addcontentsline{toc}{section}{Appendix A: Proofs}

\propConstellationSize*
\begin{proof}
Assumptions \ref{assn:demand-small} and \ref{assn:capacity-binds} ensure an interior allocation $0 < N_S < k_S$ (so $m_S > 0$), and an active rival $N_K > 0$.

The integrated firm allocates its capacity between captive launches and external sales, solving
\[
\max_{N_S,\, m_S \geq 0} \ \pi_S(N_S, m_S) = R(N_S) - c\,N_S - l_S\,(N_S + m_S) + p_{eqm}\,m_S
\qquad \text{s.t.}\qquad N_S + m_S \leq k_S,
\]
with launcher $B$ serving the residual as a monopolist at $p_{eqm} = (a + N_K - m_S + b l_B)/2b$, so $\partial p_{eqm}/\partial m_S = -1/2b$. Let $\lambda_S \geq 0$ be the multiplier on the capacity constraint. The first-order conditions for $N_S$ and $m_S$ are
\begin{align*}
R'(N_S) - c &= l_S + \lambda_S, \\
p_{eqm} - \tfrac{m_S}{2b} &= l_S + \lambda_S .
\end{align*}
Eliminating $l_S + \lambda_S$ gives the allocation condition
\[
\underbrace{R'(N_S) - c}_{\text{captive margin}} = \underbrace{p_{eqm} - \tfrac{m_S}{2b}}_{\text{external margin}},
\]
in which $l_S$ has cancelled. Constellation $K$ internalizes its own launch-price impact, $\partial p_{eqm}/\partial N_K = 1/(2b)$, so its first-order condition is $R'(N_K) - c = p_{eqm} + N_K/(2b)$. Subtracting $K$'s condition from the integrated firm's and using $R'(N_i) = \theta - \gamma(2N_i + N_j)$ with $m_S = k_S - N_S$,
\[
(1 + 2b\gamma)\,(N^{VI}_S - N^{VI}_K) = k_S ,
\]
so $N^{VI}_S - N^{VI}_K = k_S/(1 + 2b\gamma) > 0$, independent of $l_S$ and hence holding at $l_S = l_B$.
\end{proof}

\lemmOppCost*
\begin{proof}
At the joint optimum the integrated firm's two capacity first-order conditions (proof of Proposition \ref{prop:captive-expansion}) can be written so each equals the multiplier $\lambda_S$: the captive margin gives $\lambda_S = R'(N_S^*) - c - l_S$, and the external-sales margin gives $\lambda_S = (p_{eqm} - l_S) - m_S/(2b)$. Equating the two gives the stated identity.
\end{proof}

\propLaunchPrice*
\begin{proof}
With capacity binding, launcher $S$ sells its full external supply $m_S$ and cannot cover external demand, so launcher $B$ serves the remainder $X(p) + N_K - m_S$ as a monopolist, maximizing $(p - l_B)(a - bp + N_K - m_S)$. Its first-order condition gives $p^{VI}_{eqm} = (a + N_K - m_S + b l_B)/2b$. Under efficient rationing $S$'s capacity serves the top of the demand curve and $B$ is the monopoly supplier of the remainder, so $S$ cannot discipline $B$'s price once its own capacity is exhausted. At the equilibrium the integrated firm's allocation condition $R'(N_S) - c = p^{VI}_{eqm} - m_S/2b$ (equation \ref{eqn:opp-cost-engine}) holds. Rearranging it shows the same price equals $S$'s opportunity cost of an external launch, $R'(N_S) - c$, plus the residual markup $m_S/2b$. This is Lemma \ref{lemm:opportunity-cost}. In configuration $NI$ launcher $S$ is an ordinary capacity-constrained launcher and constellation $S$ an ordinary buyer, so $S$ and $K$ are symmetric independent constellations, each with effective per-satellite cost $c + p^{NI}_{eqm}$ and each internalizing its launch-price impact. Launcher $S$ again sells its full capacity $k_S$, and launcher $B$ is the residual monopolist on $X(p) + N_S + N_K - k_S$, so
\[
p^{NI}_{eqm} = \frac{a + N_S + N_K - k_S + b\,l_B}{2b},
\]
with symmetric buyer first-order conditions $R'(N_i) - c = p^{NI}_{eqm} + N_i/(2b)$ for $i \in \{S, K\}$. By symmetry $N_S^{NI} = N_K^{NI}$, which together with $R'(N_i) = \theta - \gamma(2N_i + N_j)$ pins down the common size and hence $p^{NI}_{eqm}$. Differencing against the integrated equilibrium gives $p^{VI}_{eqm} - p^{NI}_{eqm} = k_S/(6b(1+2b\gamma)) > 0$. Finally, because $N_S + m_S = k_S$, $l_S$ enters the integrated firm's profit only as the constant $l_S k_S$ and so cannot move the equilibrium price: $\partial p^{VI}_{eqm}/\partial l_S = 0$.
\end{proof}

\propLaunchEntry*
\begin{proof}
\emph{Limited capacity.} The entrant sells its full capacity $k_E$ inframarginally at the market price, leaving launcher $B$ the residual monopolist on the demand neither $S$ nor $E$ serves, $X(p) + N_K - m_S - k_E$. Maximizing $(p - l_B)(a - bp + N_K - m_S - k_E)$, its first-order condition gives the post-entry price
\[
\hat{p}_{eqm} = \frac{a + N_K - m_S - k_E + b\,l_B}{2b}.
\]
Because $E$ sells its fixed capacity $k_E$ inframarginally, its entry enters $B$'s inverse residual demand as the constant shift $a \mapsto a - k_E$, identical across configurations. The rest of the post-entry equilibrium re-solves---the integrated firm's allocation condition and $K$'s buyer condition adjust $m_S$, $N_S$, and $N_K$ to the new price---but the cross-world wedge of Proposition \ref{prop:vertical-integration-launch-price} depends only on the triple $(k_S, b, \gamma)$ and not on background demand $a$, so this constant shift leaves it unchanged:
\[
\hat{p}_{eqm}^{VI} - \hat{p}_{eqm}^{NI} = \frac{k_S}{6b(1+2b\gamma)}.
\]
The entrant earns this price on its $k_E$ launches, so its maximum sustainable fixed cost is $\tilde{F}_l = (\hat{p}_{eqm} - l_E)\,k_E$, and the integration effect is that wedge scaled by its capacity:
\[
\tilde{F}_l^{VI} - \tilde{F}_l^{NI} = (\hat{p}_{eqm}^{VI} - \hat{p}_{eqm}^{NI})\,k_E = \frac{k_E\,k_S}{6b(1+2b\gamma)} > 0.
\]
\emph{Unlimited capacity.} With $l_E < l_B$ and no capacity limit, $E$ undercuts $B$ and competes the launch price to $l_B$ in both configurations. With the price equal across configurations and the downstream allocation the same at that price, the entrant's volume and maximum sustainable fixed cost coincide, giving $\tilde{F}_l^{VI} = \tilde{F}_l^{NI}$.
\end{proof}

\propConstellationEntry*
\begin{proof}
A fresh constellation $A$ enters the services market as a third Cournot competitor alongside the captive $S$ and the independent $K$, with effective per-satellite cost $c_A + p_{eqm}$. Like $K$, it internalizes the effect of its launch demand on $p_{eqm}$.

Under $VI$ the captive $S$ is larger by Proposition \ref{prop:captive-expansion}, so by Cournot strategic substitution $A$ is smaller. Solving the three-firm equilibrium in each configuration, the capacity term $k_S$ that drives the captive expansion enters the integrated firm's allocation condition but not the non-integrated buyer's, so the same $k_S$ channel reappears:
\[
N_A^{VI} - N_A^{NI} = -\frac{k_S}{4(1+2b\gamma)} < 0 .
\]
The entrant's maximum sustainable fixed cost is its post-entry operating profit $\tilde{F}_c = R(N_A) - (c_A + p_{eqm})N_A = (P_s - c_A - p_{eqm})N_A$. Its first-order condition $R'(N_A) - c_A = p_{eqm} + N_A/(2b)$, with $R'(N_A) = P_s - \gamma N_A$, gives the margin $P_s - c_A - p_{eqm} = N_A(\gamma + 1/(2b)) = N_A(1+2b\gamma)/(2b)$, so $\tilde{F}_c = \tfrac{1+2b\gamma}{2b}\,N_A^2$, which increases in its size. Differencing across configurations and factoring the difference of squares,
\[
\tilde{F}_c^{VI} - \tilde{F}_c^{NI} = -\frac{k_S\,(N_A^{VI}+N_A^{NI})}{8b} < 0
\]
whenever $A$ is active in both configurations; the launcher cost $l_S$ does not appear. Vertical integration therefore deters constellation entry. Finally, $A$ is active under $VI$ only if it is efficient enough, $N_A^{VI} > 0$, i.e.\ $c_A < (b(2\theta + 4c - l_B) - a)/6b$; the non-integrated threshold is looser by exactly $k_S/(6b)$, so a band of entrants enter without integration but are foreclosed by it. The deterrence is strict even at that margin: at $N_A^{VI}=0$ the wedge equals $-k_S^2/(32\,b(1+2b\gamma)) < 0$.
\end{proof}

\section*{Appendix B: Endogenous
Capacity}\label{appendix-b-endogenous-capacity}
\addcontentsline{toc}{section}{Appendix B: Endogenous Capacity}

\label{app:capacity-choice}

I add a capacity-commitment stage before the market game of the body to
endogenize the integrated firm's capacity choice. In stage 1 the
integrated firm chooses launch capacity \(k_S \geq 0\) at a convex cost
\(C(k_S)\), with \(C' > 0\) and \(C'' > 0\). In stage 2 the simultaneous
game of the main text is played at the chosen \(k_S\). The continuation
is the stage-2 Nash equilibrium of that game: given \(k_S\), the
integrated firm solves \[
\pi_S^{*}(k_S) \;=\; \max_{N_S,\, m_S \geq 0}\ R(N_S) - c\,N_S - l_S\,(N_S + m_S) + p_{eqm}\,m_S
\quad\text{s.t.}\quad N_S + m_S \leq k_S,
\] with launcher \(B\) pricing the residual at
\(p_{eqm} = (a + N_K - m_S + b\,l_B)/2b\) and constellation \(K\)
best-responding, as in the proof of Proposition
\ref{prop:captive-expansion}. Write \(\lambda_S(k_S) \geq 0\) for the
multiplier on the capacity constraint at that equilibrium. Stage 1 then
solves \[
\max_{k_S \geq 0}\ \Pi_S(k_S) \;=\; \pi_S^{*}(k_S) - C(k_S).
\]

The integrated firm builds capacity to the point where the marginal
value of a unit equals its marginal cost. The marginal value is the
capacity rent \(\lambda_S(k_S)\), the multiplier on the stage-2
constraint \(N_S + m_S \leq k_S\); by complementary slackness it is
positive exactly when the firm uses its full capacity. Let
\(k_S^{\mathrm{sat}}\) denote the firm's unconstrained stage-2 launch
demand, the total launches \(N_S + m_S\) it would choose were capacity
not a constraint. For \(k_S \geq k_S^{\mathrm{sat}}\) the stage-2
constraint is slack and \(\lambda_S = 0\); for
\(k_S < k_S^{\mathrm{sat}}\) it binds, and \(\lambda_S(k_S)\) falls to
zero as \(k_S \to k_S^{\mathrm{sat}}\). The firm builds where this rent
equals the marginal cost of capacity, the peak-load investment rule of
Williamson (\citeproc{ref-williamson1966}{1966}) and Chao
(\citeproc{ref-chao1983}{1983}). Because capacity is costly
(\(C' > 0\)), it stops at a fleet \(k_S^{*} < k_S^{\mathrm{sat}}\) whose
rent is still positive. It therefore operates with capacity fully
utilized, producing the binding regime analyzed in the main text
(Assumption \ref{assn:capacity-binds}).
Figure~\ref{fig-capacity-binding} illustrates the intuition and
Proposition \ref{prop:capacity-binds} formalizes.

\begin{figure}

\centering{

\includegraphics[width=0.75\linewidth,height=\textheight,keepaspectratio]{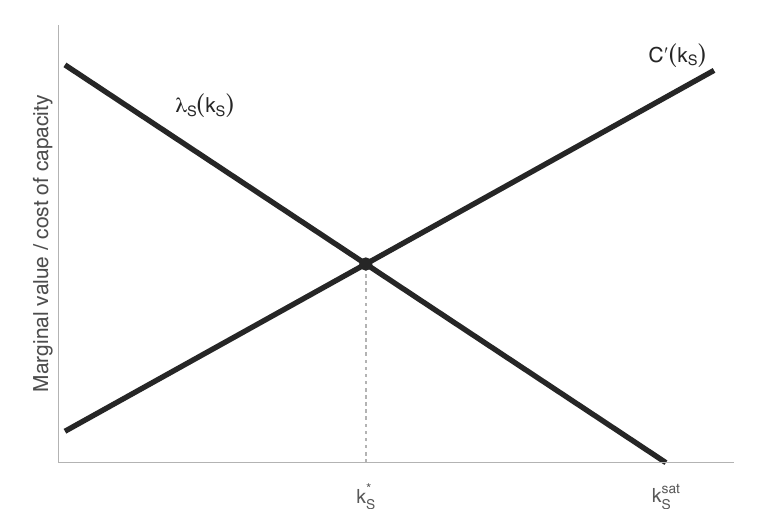}

}

\caption{\label{fig-capacity-binding}\textbf{The optimal capacity
choice.} The integrated firm builds where the capacity rent
\(\lambda_S(k_S)\), the marginal value of a unit of capacity, equals its
marginal cost \(C'(k_S)\). The rent falls to zero at
\(k_S^{\mathrm{sat}}\), the firm's unconstrained stage-2 launch demand;
because \(C' > 0\) the optimum \(k_S^{*}\) lies below
\(k_S^{\mathrm{sat}}\), where the rent remains positive and the capacity
constraint binds. \emph{Demand-side parameters} \(\theta=8\), \(a=11\),
\(\gamma=0.1\), \(c=0.5\), \(l_S=1\), \(l_B=2\), \(b=1\);
\emph{illustrative capacity cost} \(C'(k_S)=w+d\,k_S\) with \(w=0.5\),
\(d=0.14\).}

\end{figure}%

\begin{prop}[Optimal capacity binds]\label{prop:capacity-binds}
Let the integrated firm choose capacity in stage~1 at a convex cost $C$ with $C' > 0$, $C'' > 0$, and $C'(0) < \lambda_S(0)$. Its optimal capacity $k_S^{*}$ is unique and interior, and satisfies
\[
\lambda_S(k_S^{*}) = C'(k_S^{*}) > 0 .
\]
The capacity constraint binds with strictly positive shadow value, and $k_S^{*}$ lies below the satiation level $k_S^{\mathrm{sat}}$ at which $\lambda_S(k_S^{\mathrm{sat}}) = 0$.
\end{prop}

\begin{proof}
Capacity enters the stage-2 program only through the constraint $N_S + m_S \leq k_S$, so by the envelope theorem $\pi_S^{*\prime}(k_S) = \lambda_S(k_S)$, the constraint's multiplier. The stage-2 equilibrium (proof of Proposition~\ref{prop:captive-expansion}) gives
\[
\lambda_S(k_S) = \frac{\theta - c - l_S + \gamma\,(a - 2k_S + b\,l_B - 2b\,l_S)}{1 + 2b\gamma},
\]
strictly decreasing in $k_S$. Hence $\Pi_S'(k_S) = \lambda_S(k_S) - C'(k_S)$, with $\Pi_S'(0) = \lambda_S(0) - C'(0) > 0$ and $\Pi_S''(k_S) = \lambda_S'(k_S) - C''(k_S) < 0$. So $\Pi_S$ is strictly concave, and its maximizer is unique and interior, solving $\lambda_S(k_S^{*}) = C'(k_S^{*})$. Since $C' > 0$, $\lambda_S(k_S^{*}) > 0$, so by complementary slackness the stage-2 constraint binds. Because $\lambda_S$ is decreasing with $\lambda_S(k_S^{\mathrm{sat}}) = 0 < C'(k_S^{\mathrm{sat}})$, the optimum satisfies $k_S^{*} < k_S^{\mathrm{sat}}$.
\end{proof}

\section*{Appendix C: Falcon 9 Advertised
Prices}\label{appendix-c-falcon-9-advertised-prices}
\addcontentsline{toc}{section}{Appendix C: Falcon 9 Advertised Prices}

\label{app:falcon-prices}

The advertised (``sticker'') Falcon 9 prices used in the introduction
and in Figure \ref{fig-stylized-facts} are the prices listed on SpaceX's
published rate card in nominal U.S. dollars.
Table~\ref{tbl-falcon-prices} shows the annual price snapshots with
sources. All prices are for the standard, single-launch payment plan to
a reference orbit; multi-launch contracts, rideshare, and government
missions are priced separately and are not shown. Actual prices paid by
customers for single-launch deliveries may vary by mission parameters,
e.g., whether the booster is recovered or not.

\begin{longtable}[]{@{}
  >{\raggedright\arraybackslash}p{(\linewidth - 4\tabcolsep) * \real{0.3333}}
  >{\raggedright\arraybackslash}p{(\linewidth - 4\tabcolsep) * \real{0.3333}}
  >{\raggedright\arraybackslash}p{(\linewidth - 4\tabcolsep) * \real{0.3333}}@{}}
\caption{Falcon 9 advertised single-launch prices and their archived
sources.}\label{tbl-falcon-prices}\tabularnewline
\toprule\noalign{}
\begin{minipage}[b]{\linewidth}\raggedright
Year
\end{minipage} & \begin{minipage}[b]{\linewidth}\raggedright
Advertised price
\end{minipage} & \begin{minipage}[b]{\linewidth}\raggedright
Source
\end{minipage} \\
\midrule\noalign{}
\endfirsthead
\toprule\noalign{}
\begin{minipage}[b]{\linewidth}\raggedright
Year
\end{minipage} & \begin{minipage}[b]{\linewidth}\raggedright
Advertised price
\end{minipage} & \begin{minipage}[b]{\linewidth}\raggedright
Source
\end{minipage} \\
\midrule\noalign{}
\endhead
\bottomrule\noalign{}
\endlastfoot
2012--2013 & \$54M & Internet Archive capture
\href{https://web.archive.org/web/20130715094112/http://www.spacex.com/falcon9}{2013-07-15} \\
2014--2016 & \$61.2M &
\href{https://spacenews.com/spacexs-reusable-falcon-9-what-are-the-real-cost-savings-for-customers/}{\emph{SpaceNews}}
(Apr.~2016) \\
2019--2021 & \$62M & Internet Archive captures
\href{https://web.archive.org/web/20200527123302/https://www.spacex.com/media/Capabilities\%26Services.pdf}{2020-05-27},
\href{https://web.archive.org/web/20201101010603/https://www.spacex.com/media/Capabilities\%26Services.pdf}{2020-11-01} \\
2022--2023 & \$67M & Internet Archive captures
\href{https://web.archive.org/web/20220322170331/https://www.spacex.com/media/Capabilities\%26Services.pdf}{2022-03-22},
\href{https://web.archive.org/web/20230331145041/https://www.spacex.com/media/Capabilities\%26Services.pdf}{2023-03-31} \\
2024 & \$69.75M & Internet Archive capture
\href{https://web.archive.org/web/20240607073822/https://www.spacex.com/media/Capabilities\%26Services.pdf}{2024-06-07} \\
2025 & \$69.85M & Internet Archive capture
\href{https://web.archive.org/web/20250511135440/https://www.spacex.com/media/Capabilities\%26Services.pdf}{2025-05-11} \\
2026 & \$74M & SpaceX \emph{Capabilities \& Services} rate card
(\href{https://www.spacex.com/media/Capabilities\%26Services.pdf}{current}) \\
\end{longtable}

\noindent The headline comparison in the introduction uses the 2012 and
2026 endpoints, the latter deflated to 2012 dollars with the NASA New
Start Inflation Index.

\end{document}